\documentclass[a4paper,UKenglish,cleveref, autoref]{lipics-v2019}

\nolinenumbers

\title{On Maximum-Sum Matchings of Points}

\titlerunning{On Maximum-Sum Matchings of Points}

\author{Sergey Bereg}{Department of Computer Science, University of Texas at Dallas, USA}{besp@utdallas.edu}{}{Supported in part by NSF award CCF-1718994}

\author{Oscar Chac\'on-Rivera}{Pontificia Universidad Cat\'olica de Chile, Chile}{opchacon@uc.cl}{}{}

\author{David Flores-Pe\~naloza}{Departamento de Matem\'aticas, Facultad de Ciencias, Universidad Nacional Aut\'onoma de M\'exico, D.F. M\'exico, Mexico}{dflorespenaloza@gmail.com}{}{Partially supported by PAPIIT IN117317 (UNAM, Mexico)}

\author{Clemens Huemer}{Departament de Matem\`{a}tiques, Universitat Polit\`{e}cnica de Catalunya, Spain}{clemens.huemer@upc.edu}{}{Partially supported by projects MTM2015-63791-R (MINECO/FEDER) and Gen.\ Cat.\ DGR 2017SGR1336}

\author{Pablo P\'erez-Lantero}{Departamento de Matem\'atica y Ciencia de la Computaci\'on, Universidad de Santiago, Chile}{pablo.perez.l@usach.cl}{}{Partially supported by projects DICYT 041933PL Vicerrector\'ia de Investigaci\'on, Desarrollo e Innovaci\'on USACH (Chile), and Programa Regional STICAMSUD 19-STIC-02}

\author{Carlos Seara}{Departament de Matem\`{a}tiques, Universitat Polit\`{e}cnica de Catalunya, Spain}{carlos.seara@upc.edu}{}{Supported by projects MTM2015-63791-R MINECO/FEDER and Gen.~Cat.\ DGR 2017SGR1640}

\authorrunning{S.\ Bereg et al.}

\Copyright{Sergey Bereg, Oscar Chac\'on-Rivera, David Flores-Pe\~naloza, Clemens Huemer, Pablo P\'erez-Lantero, and Carlos Seara}



\SeriesVolume{xx}
\ArticleNo{xx}

\newcommand{\M}{\ensuremath{\mathcal{M}}}
\newcommand{\E}{\ensuremath{\mathcal{E}}}
\newcommand{\cost}{\ensuremath{\operatorname{cost}}}

\newcommand\blfootnote[1]{%
  \begingroup
  \renewcommand\thefootnote{}\footnote{#1}%
  \addtocounter{footnote}{-1}%
  \endgroup}

\begin{document}
\maketitle

%

\begin{abstract}
Huemer et al.\ (Discrete Mathematics, 2019) proved that for any two point sets $R$ and $B$ with $|R|=|B|$, the perfect matching that matches points of $R$ with points of $B$, and maximizes the total \emph{squared} Euclidean distance of the matched pairs, verifies that all the disks induced by the matching have a common point. Each pair of matched points $p\in R$ and $q\in B$ induces the disk of smallest diameter that covers $p$ and $q$. Following this research line, in this paper we consider the perfect matching that maximizes the total Euclidean distance. First, we prove that this new matching for $R$ and $B$ does not always ensure the common intersection property of the disks. Second, we extend the study of this new matching for sets of $2n$ uncolored points in the plane, where a matching is just a partition of the points into $n$ pairs. As the main result, we prove that in this case all disks of the matching do have a common point. This implies a big improvement on a conjecture of Andy Fingerhut in 1995, about a maximum matching of $2n$ points in the plane.
\end{abstract}

\section{Introduction}


Let $R$ and $B$ be two disjoint point sets in the plane with $|R|=|B|=n$, $n\ge 2$. The points in $R$ are \emph{red}, and those in $B$ are \emph{blue}. A \emph{matching} of $R\cup B$ is a partition of $R\cup B$ into $n$ pairs such that each pair consists of a red and a blue point. A point $p\in R$ and a point $q\in B$ are \emph{matched} if and only if the (unordered) pair $(p,q)$ is in the matching. For every $p,q\in\mathbb{R}^2$, we use $pq$ to denote the segment connecting $p$ and $q$, and $\|p-q\|$ to denote its length, which is the Euclidean norm of the vector $p-q$. Let $D_{pq}$ denote the disk with diameter equal to $\|p-q\|$, that is centered at the midpoint $\frac{p+q}{2}$ of the segment $pq$. For any matching $\M$, we use $D_\M$ to denote the set of the disks associated with the matching, that is, $D_\M=\{D_{pq} \mid (p,q) \in\M\}$.


Huemer et al.~\cite{HUEMER2019} proved that if $\M$ is any matching that maximizes the total \emph{squared} Euclidean distance of the matched points, i.e., it maximizes $\sum_{(p,q)\in\M}\|p-q\|^2$, then all disks of $D_\M$ have a point in common.

In this paper, we will consider the \emph{max-sum} matching $\E$, as the matching that maximizes the total Euclidean distance of the matched points. For any matching $\M$, let $\cost(\M)$ denote the sum $\sum_{(p,q)\in\M}\|p-q\|$. Thus, $\E$ is such that $\cost(\E)$ is maximum among all matchings. In Section~\ref{sec2}, we prove that every pair of disks in $\E$ have a common point, but it cannot be guaranteed that \emph{all disks} have a common point. 

In this paper, we will also consider max-sum matchings of sets of $2n$ uncolored points in the plane, where a matching is just a partition of the points into $n$ pairs. Then, we study the following problem: Fingerhut~\cite{andy} conjectured that given a set $P$ of $2n$ uncolored points in the plane and a max-sum matching $\{(a_i,b_i),i=1,\dots,n\}$ of $P$, there exists a point $o$ of the plane, not necessarily a point of $P$, such that
\begin{equation}\label{eqFingerhut}
    \|a_i-o\|+\|b_i-o\|~\le~ \frac{2}{\sqrt{3}}\cdot \|a_i-b_i\| \hspace{0.3cm} \textrm{for all } i\in\{1,\ldots,n\}, ~\textrm{where } 2/\sqrt{3}\approx 1.1547.
\end{equation}
The statement of equation~\eqref{eqFingerhut} is equivalent to stating that the intersection $E_1\cap E_2\cap\dots\cap E_n$ is not empty, where $E_i$ is the region bounded by the ellipse with foci $a_i$ and $b_i$, and semimajor axis length $(1/\sqrt{3})\cdot \|a_i-b_i\|$, for all $i\in\{1,\ldots,n\}$~\cite{andy}. Then, by Helly's Theorem, it is sufficient to prove equation~\eqref{eqFingerhut} for $n=3$. As noted by Andy Fingerhut, the factor $2/\sqrt{3}$ is the minimum possible. It is enough to consider an equilateral triangle, where at each vertex two points are located. The max-sum matching of the six points is made of pairs of vertex-opposed points, and the regions defined by the three ellipses using the factor $2/\sqrt{3}$ have exactly one point in common. 

Andy Fingerhut was also interested in a small constant close to $2/\sqrt{3}$, and Eppstein~\cite{andy} proved that the result holds with $2.5$ instead of $2/\sqrt 3$. The proof is simple and does not require Helly's theorem. Let $o$ be the midpoint of the shortest edge in the matching. Namely, for all $i\in\{1,\ldots,n\}$, 
\begin{equation}\label{eqEppstein}
    \|a_i-o\|+\|b_i-o\|~\le~ 2.5\cdot \|a_i-b_i\|.
\end{equation}

In Section~\ref{sec3}, we first show that the constant in this inequality can be improved to $\sqrt 5\approx 2.236$ by using the same point $o$, but refined arguments. Second, as the main result of this paper, we improve even further this constant. Precisely, we prove that for any point set $P$ of $2n$ uncolored points in the plane and a max-sum matching $\E=\{(a_i,b_i),i=1,\dots,n\}$ of $P$, all disks in $D_{\E}$ have a common intersection. This directly implies that any point $o$ in the common intersection satisfies
\[
	\|a_i-o\|+\|b_i-o\|~\le~ \sqrt{2}\cdot  \|a_i-b_i\|
\]
for all $i\in\{1,\ldots,n\}$, where $\sqrt{2}\approx 1.4142$. We will use Helly's theorem, that is, we will prove the claim for $n=3$.

{\bf Remark}. 
Recently, Adiprasito \textit{et al}.~\cite{ABMT} proved a no-dimension version of Tverberg's theorem:
For any set $P$ of $n$ points in $\mathbb{R}^d$ and $k\in\{2,\ldots,n\}$, there exists a point $o \in \mathbb{R}^d$ and a partition of $P$ into $k$ sets $P_1,\ldots, P_k$ such that, for all $i\in\{1,\ldots,k\}$,
$
d(o,conv(P_i)) \leq (2+\sqrt{2})\sqrt{\frac{k}{n}}~diam(P),
$
where $conv(P_i)$ is the convex hull of $P_i$, $d(o,conv(P_i))$ is the Euclidean distance from $o$ to $conv(P_i)$, and $diam(P)$ is the diameter of $P$. Our result can improve the bound to $d(o,conv(P_i)) \leq \frac{1}{2}~diam(P_i) \leq \frac{1}{2}~diam(P)$ for $d=2$ and $n=2k$ (using the max-sum matching for $P$ and taking $o$ in the intersection of the disks).  For $d=2$ and arbitrary values of $k$, another bound of $d(o,conv(P_i)) \leq \frac{1}{\sqrt{3}}~diam(P) $ is implied by Jung's theorem~\cite{jung-1910} (taking as $o$ the center of the smallest disk enclosing $P$).

\subsection{Related problems}\label{sec1.1}

As described by Huemer et al.~\cite{HUEMER2019}, this class of problems is well studied in discrete and computational geometry, starting from the classic result that $n$ red points and $n$ blue points can always be perfectly matched with $n$ pairwise non-crossing segments, where each segment connects a red point with a blue point~\cite{larson1983problem}. The study has been continued in plenty of directions, for both the monochromatic and bichromatic versions, by using pairwise disjoint objects inducing the matching: segments~\cite{aloupis2015,dumitrescu2001}, rectangles~\cite{abrego2004matching,abrego2009,bereg2009,caraballo2014matching,Corujo2019}, and more general geometric objects~\cite{AloupisCCDDDMHHLSST13}. 

Our results, as that of~\cite{HUEMER2019}, are in the direction opposite to that of above mentioned known results on matching points: The goal is that all matching objects (i.e., the disks in $D_{\M}$) have a common intersection, whereas in previous work it is required that all matching objects are pairwise disjoint.

The problem of stabbing a finite set of pairwise intersecting disks with four or five points~\cite{carmi2019,danzer1986losung,har2019stabbing}, is also related with this paper, since we deal with finite sets of pairwise intersecting disks.

Particularly, the conjecture posted by Fingerhut~\cite{andy} relates to an abstract version of a problem in designing communication networks. The given $2n$ points represent nodes that should be connected by a network. The cost $\cost(\E)$ of a max-sum matching $\E$ is a lower bound on the cost of any feasible network, and the point $o$ is a place where one would like to place the center of a ``star'' network, in which all given points are connected directly to the center of the star.

\section{Red-blue matchings}\label{sec2}

Let $R$ and $B$ be two disjoint point sets defined as above, where $|R|=|B|=n$, $n\ge 2$. In this section, we prove that the perfect matching $\E$ of $R$ and $B$ that maximizes the total Euclidean distance $\operatorname{cost}(\E)$ does not always ensure the common intersection property of the disks in $D_\E$. Nevertheless, although the common intersection is not always possible, all disks must be pairwise intersecting, as proved in the next proposition.

\begin{proposition}\label{prop1}
	Every pair of disks in $D_\E$ have a non-empty intersection.
\end{proposition}

\begin{proof}
Let $(a,a')$ and $(b,b')$ be two different pairs of $\E$, with $a,b\in R$ and $a',b'\in B$. Since $\E$ is a maximum matching, we must have $\|a-b'\|+\|a'-b\|\le \|a-a'\| + \|b-b'\|$. The equality holds, for instance, when $a'$, $a$, $b$, and $b'$ are in this order consecutive vertices of a square. Note that, in the contrary case, we will have $\cost((\E\setminus\{(a,a'),(b,b')\})\cup\{(a,b'),(b,a')\})>\cost(\E)$.

The disks $D_{aa'},D_{bb'}\in D_\E$ have a common point if and only if the distance $\left\|\frac{a+a'}{2}-\frac{b+b'}{2}\right\|$ between the centers $\frac{a+a'}{2}$ and $\frac{b+b'}{2}$ of $D_{aa'}$ and $D_{bb'}$, respectively, is at most the sum $\frac{\|a-a'\|}{2} + \frac{\|b-b'\|}{2}$ of their radii. The following equations ensure this condition:
\begin{eqnarray*}
	\left\|(a+a')-(b+b')\right\|
		& = & \left\|(a-b')+(a'-b)\right\| \\
		& \le & \|a-b'\|+\|a'-b\| \\
		& \le & \|a-a'\| + \|b-b'\|.
\end{eqnarray*}
Hence, we conclude that $D_{aa'}\cap D_{bb'}\neq\emptyset$ for every pair of disks $D_{aa'}$ and $D_{bb'}$ of $D_\E$.
\end{proof}

\begin{theorem}\label{theo:not3}
There exist point sets $R\cup B$, with $R\cap B=\emptyset$ and $|R|=|B|=3$, such that, for any max-sum matching $\E$ of $R$ and $B$, the intersection of the disks of $D_{\E}$ is the empty set.
\end{theorem}

\begin{proof}
Let $R=\{a,b,c\}$ and $B=\{a',b',c'\}$, with $a=(-1,0)$, $b=(1,0)$, $c=(0,\sqrt{3})$, $c'=(0,3)$, and $a'\in bc$ and $b'\in ac$ such that $\|c-a'\|=\|c-b'\|=\varepsilon$, for a parameter $\varepsilon>0$ that ensures that $\E=\{(a,a'),(b,b'),(c,c')\}$ is the only maximum matching of $R\cup B$ (see Figure~\ref{fig:TheoNot3}).

\begin{figure}[t]
	\centering
	\includegraphics[scale=1.0,page=1]{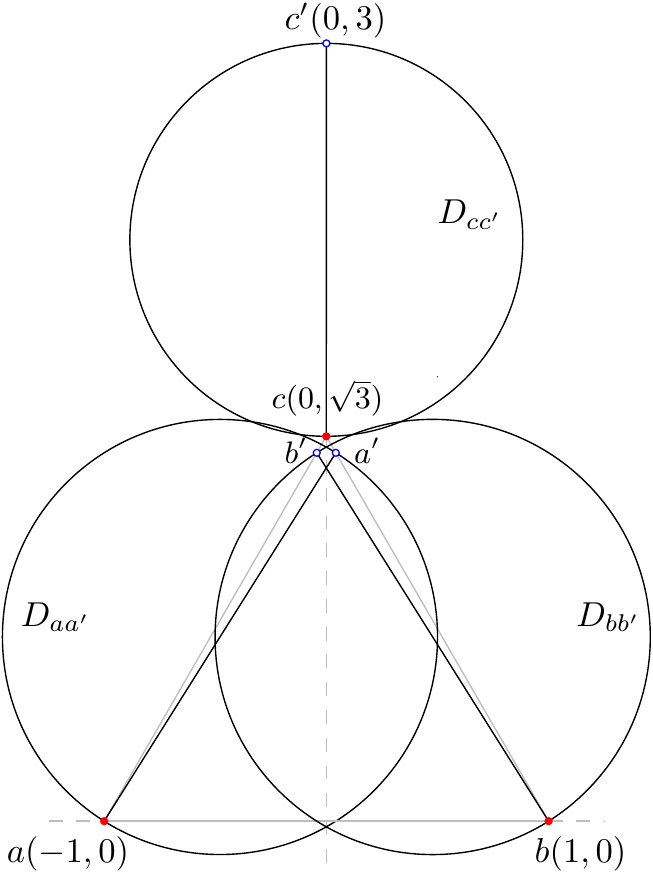}
	\caption{\small{Proof of Theorem~\ref{theo:not3}.}}
	\label{fig:TheoNot3}
\end{figure}

Note that
\begin{eqnarray*}
	\|a-b'\|+\|b-c'\|+\|c-a'\| 
		& = & \|a-c'\|+\|b-a'\|+\|c-b'\| \\
		& = & \sqrt{10} + (2-\varepsilon) + \varepsilon \\
		& = & 2+\sqrt{10},
\end{eqnarray*}
and we need to ensure that
\[
	2+\sqrt{10} ~<~ \|a-a'\|+\|b-b'\|+\|c-c'\| ~=~ \cost(\E).
\]
That is, the matching $\{(a,a'),(b,b'),(c,c')\}$ has larger total Euclidean distance than the matchings $\{(a,b'),(b,c'),(c,a')\}$ and $\{(a,c'),(b,a'),(c,b')\}$. Since
\begin{eqnarray*}
	\|a-a'\|+\|b-b'\|+\|c-c'\| 
		& = & 2\|a-a'\|+\|c-c'\| \\
		& = & 2\|a-a'\|+(3-\sqrt{3}) \\
		& > & 2(\|a-c\|-\varepsilon)+ (3-\sqrt{3}) \\
		& = & 7 - \sqrt{3} - 2\varepsilon,
\end{eqnarray*}
it suffices to ensure
\begin{equation}\label{eq1}
	2+\sqrt{10} ~<~ 7 - \sqrt{3} - 2\varepsilon ~~\Longleftrightarrow~~ \varepsilon ~<~ \frac{5-\sqrt{10}-\sqrt{3}}{2} ~\approx~ 0.0528.
\end{equation}
Furthermore, since
\begin{eqnarray*}
	\|a-c'\|+\|b-b'\|+\|c-a'\| 
	& = & \|a-a'\|+\|b-c'\|+\|c-b'\| \\
	& < & (2+\varepsilon) + \sqrt{10} + \varepsilon \\
	& = & 2+\sqrt{10}+2\varepsilon,
\end{eqnarray*}
to ensure that $\{(a,a'),(b,b'),(c,c')\}$ has larger total Euclidean distance than $\{(a,c'),(b,b'),(c,a')\}$ and $\{(a,a'),(b,c'),(c,b')\}$, it suffices to guarantee that
\begin{equation}\label{eq2}
	2+\sqrt{10} +2\varepsilon ~<~ 7 - \sqrt{3} - 2\varepsilon ~~\Longleftrightarrow~~ \varepsilon ~<~ \frac{5-\sqrt{10}-\sqrt{3}}{4} ~\approx~ 0.0264.
\end{equation}
Hence, any $\varepsilon>0$ satisfying~\eqref{eq2} (and then also~\eqref{eq1}) is such that $\E=\{(a,a'),(b,b'),(c,c')\}$ is the only maximum matching of $R\cup B$. It remains to show that $D_{aa'}\cap D_{bb'}\cap D_{cc'}=\emptyset$. To see that, it is straightforward to show (based on the fact that $c\notin D_{aa'}$ and $c\notin D_{bb'}$) that all points of $D_{aa'}\cap D_{cc'}$ have negative $x$-coordinates, and all points of $D_{bb'}\cap D_{cc'}$ have positive $x$-coordinates.
\end{proof}

Let $|R|=|B|=n$. If $n=2$, then the disks of $D_{\E}$ always intersect by Proposition~\ref{prop1}. If $n=3$, then the disks of $D_{\E}$ may not intersect by Theorem~\ref{theo:not3}. One could expect that the disks of $D_{\E}$ always intersect if $n$ is sufficiently large. Next, we answer this question negatively.

\begin{theorem}\label{theo:not4}
For any $n\ge 4$, there exist point sets $R\cup B$, with $R\cap B=\emptyset$ and $|R|=|B|=n$, such that, for any max-sum matching $\E$ of $R$ and $B$, the intersection of the disks of $D_{\E}$ is the empty set.
\end{theorem}

\begin{proof}
We construct a set $R$ of $n$ red points, and a set $B$ of $n$ blue points, as follows. First, take six points $a,b,c\in R$ and $a',b',c'\in B$ as in Theorem~\ref{theo:not3}. Second, we add $n-3$ red points and $n-3$ blue points in the $\varepsilon$-neighborhood of point $c$, as explained below (see Figure~\ref{fig:TheoNot4}, a zoomed-in view of Figure~\ref{fig:TheoNot3}), where $\varepsilon>0$ is a sufficiently small number that will be specified later. 

\begin{figure}[t]
	\centering
	\includegraphics[scale=1.0,page=2]{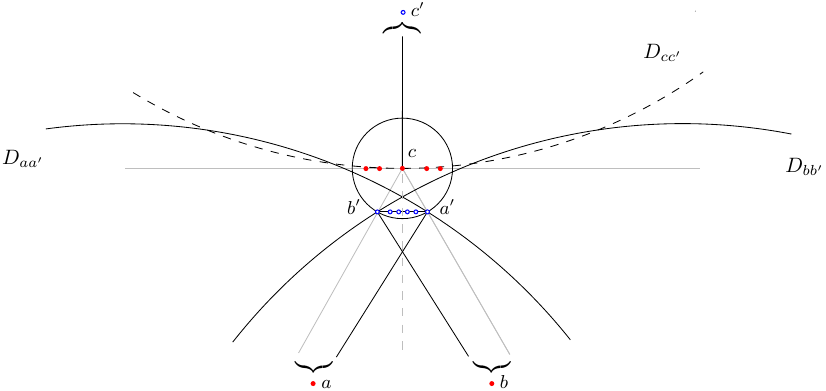}
	\caption{\small{Proof of Theorem~\ref{theo:not4}.}}	
	\label{fig:TheoNot4}
\end{figure}

Add $n-3$ blue points, denoted $a'_1,a'_2,\dots,a'_{n-3}$, on the segment $b'a'$. Note that they are within distance $\varepsilon$ from $c$ since $\|c-b'\|=\|c-a'\|=\varepsilon$. Add $n-3$ red points, denoted $a_1,a_2,\dots,a_{n-3}$, on the horizontal line through $c$ (perpendicular to $cc'$), and within distance $\varepsilon$ from $c$ (see Figure~\ref{fig:TheoNot4}). Then, we have $\|c'-a_i\|\ge\|c'-c\|$ for all $i\in\{1,\ldots,n-3\}$. 

Consider two matchings $\M_1$ and $\M_2$ such that: in $\M_1$ point $c'$ is matched to $c$ or to some $a_i$, $i\in\{1,\dots,n-3\}$; and in $\M_2$ point $c'$ is matched to $a$ or to $b$. Given $p\in R\cup B$, and a matching $\M$ of $R\cup B$, let $\M(p)$ be the point such that $p$ and $\M(p)$ are matched in $\M$. Clearly, it holds that
\[
    \|c'-\M_1(c')\|~\ge~\|c'-c\|~=~3-\sqrt{3}~~\text{and}~~\|a-\M_1(a)\|~\ge~ \|a-c\|-\varepsilon~=~2-\varepsilon.
\]
Similarly, $\|b-\M_1(b)\|\ge 2-\varepsilon$. Then, $\cost(\M_1)\ge 3-\sqrt{3}+2(2-\varepsilon)=7 - \sqrt{3} - 2\varepsilon$.

By symmetry, we can assume that $(a,c')\in\M_2$. In $\M_2$, point $b$ is matched to $a'$, to $b'$, or to some $a'_i$, $i\in\{1,\dots,n-3\}$. Then, we have that
\[
    \|b-\M_2(b)\|~\le~\|b-b'\| ~\le~ \|b-c\|+\varepsilon ~=~ 2+\varepsilon.
\]
Furthermore, the members of the remaining $n-2$ pairs of $\M_2$ are within distance $\varepsilon$ from $c$. Thus, $\cost(\M_2)\le \sqrt{10}+2+\varepsilon+2(n-2)\varepsilon$. We choose $\varepsilon>0$ such that
\[
	7 - \sqrt{3} - 2\varepsilon ~>~  \sqrt{10}+2+\varepsilon+2(n-2)\varepsilon,
\]
or what is the same,
\[
	5 - \sqrt{3} - \sqrt{10} ~>~ (2n-1) \varepsilon.
\]
Since $5-\sqrt{3}-\sqrt{10}\approx 0.10567$, we can choose $\varepsilon<\frac{1}{10(2n-1)}$ to ensure $\cost(\M_1)>\cost(\M_2)$.

Therefore, $\E \ne \M_2$ and we can assume that $\E=\M_1$. To show that the intersection of the disks in $D_{\E}$ is empty, it suffices to show that the disks $D_{a\M_1(a)},D_{b\M_1(b)}$ and $D_{c'\M_1(c')}$ do not have a common point. To this end, note that (which is straightforward to prove by construction) that all points of $D_{a\M_1(a)}\cap D_{c'\M_1(c')}$ have negative $x$-coordinates, and all points of $D_{b\M_1(b)}\cap D_{c'\M_1(c')}$ have positive $x$-coordinates. The result thus follows.
\end{proof}

\section{On Fingerhut's Conjecture}\label{sec3}

In this section, we first show that the constant in inequality~\eqref{eqEppstein} can be improved to $\sqrt 5\approx 2.236$. We start by proving the next technical lemma:

\begin{lemma}\label{lem:sqrt5}
Let $p$ and $q$ be two points of the plane, and let $r_{pq}$ be the radius of disk $D_{pq}$. Let $D$ be a second disk with center $o$ and radius $r\le r_{pq}$ such that $D\cap D_{pq}\neq \emptyset$. Then,
\[
	\|p-o\|+\|q-o\| ~\le~ \sqrt{5}\cdot \|p-q\|.
\]
\end{lemma}

\begin{proof}
Let $o_{pq}$ denote the center of $D_{pq}$, and note that $\|o-o_{pq}\|\le r+r_{pq}\le 2\cdot r_{pq}=\|p-q\|$, because $D\cap D_{pq}\neq \emptyset$. Then, we have that:
\begin{align*}
	(\|p-o\|+\|q-o\|)^2 & ~\le~ 2\cdot (\|p-o\|^2+\|q-o\|^2) & \text{(Cauchy-Schwarz's inequality)} \\
	        & ~=~ 2\left(\frac{1}{2}\|p-q\|^2+2\cdot \|o-o_{pq}\|^2\right) & \text{(Apollonius' theorem)} \\
			& ~=~ \|p-q\|^2+4\cdot \|o-o_{pq}\|^2 \\
			& ~\le~ 5\cdot \|p-q\|^2,
\end{align*}
which implies the lemma.
\end{proof}

Lemma~\ref{lem:sqrt5} and Proposition~\ref{prop1} imply the following theorem.

\begin{theorem}\label{theo:sqrt5}
For any $n\ge 2$, and any planar $n$-point sets $R$ and $B$ such that $R\cap B=\emptyset$, consider any max-sum matching $\E=\{(a_i,b_i), i=1,\dots,n\}$ of $R\cup B$.  Let $o$ be the midpoint of the shortest segment in the matching $\E$. Then, for all $i\in\{1,\ldots,n\}$, 
\[
    \|a_i-o\|+\|b_i-o\|~\le~ \sqrt 5\cdot \|a_i-b_i\|.
\]
\end{theorem}

Note that the bound of $\sqrt{5}$ is tight if the point $o$ is always the midpoint of the shortest segment. Namely, consider 2 red points and 2 blue points, as vertices of a square, such that diagonal-opposed vertices have the same color.

We now strengthen Theorem~\ref{theo:sqrt5}. Precisely, we prove that for any point set $P$ of $2n$ uncolored points in the plane and a max-sum matching $\M=\{(a_i,b_i),i=1,\dots,n\}$ of $P$, all disks in $D_{\M}$ have a common intersection. We will use Helly's theorem, that is, we will prove the claim for $n=3$.

We note that $P$ might contain different points with the same coordinates. Furthermore, the common intersection might be a singleton. Namely, consider six points $a$, $b$, $c$, $a'$, $b'$, and $c'$, where $a$, $b$, and $c$ are the vertices of a non-empty triangle, and $a'$, $b'$, and $c'$ coincide with a point $z$ in the interior of the triangle. By the triangle inequality, the matching $\{(a,a'),(b,b'),(c,c')\}$ is max-sum, and $z$ is the only point in the common intersection $D_{aa'}\cap D_{bb'}\cap D_{cc'}$. For any matching $\M$, the {\em segments} of $\M$ is the segment set $\{pq\mid (p,q)\in\M\}$.

Let $p$, $q$, and $r$ be three points of the plane. We denote by $\Delta pqr$ the triangle with vertices $p$, $q$, and $r$; by $\ell(p,q)$ the straight line through $p$ and $q$ oriented from $p$ to $q$; by $\tau(p,q)$ the ray with apex $p$ that goes through $q$; by $\vec{pq}$ the segment $pq$ oriented from $p$ to $q$; and by $C_{pq}$ the circle bounding $D_{pq}$. Furthermore, given a fourth point $s$, we say that $\vec{pq}$ {\em points to} $rs$ if $q$ is in the interior of $\Delta prs \cap D_{rs}$. See Figure~\ref{fig:lem:point-in-disk} (left), where segment $\vec{cd}$ points to $ab$.

The first observation that we state is that if $\M$ is a max-sum matching, then the disks of $D_{\M}$ intersect pairwise, as a consequence of Proposition~\ref{prop1}. Second, every pair of segments of $\M$ either cross or one oriented segment points to the other one, as formally proved in the following lemma. Third, the four vertices of any two segments cannot be in convex position, because the matching would not be max-sum by the triangle inequality.  

\begin{lemma}\label{lem:point-in-disk}
Let $\{a,b,c,d\}$ be a set of four points such that $\{(a,b),(c,d)\}$ is a max-sum matching of $\{a,b,c,d\}$ and $d$ belongs to the interior of $\Delta abc$. Then, $d$ belongs to the interior of disk $D_{ab}$. That is, $\vec{cd}$ points to $ab$.
\end{lemma}

\begin{proof}
Let $d'$ be the reflection of $d$ about the midpoint $m=(a+b)/2$ of segment $ab$, also the center of $D_{ab}$, and assume w.l.o.g.\ that $d$ belongs to triangle $\Delta acd'$. Refer to Figure~\ref{fig:lem:point-in-disk}. Then,
\[
	\|c-d\| + \|d-d'\| ~<~ \|c-a\| + \|a-d'\|.
\]
Namely, the perimeter of triangle $\Delta dcd'$ is smaller than the perimeter of triangle $\Delta acd'$, which implies the equation. Using that $\|a-d'\|=\|d-b\|$ because the quadrilateral with vertex set $\{a,d',b,d\}$ is a parallelogram (or rhomboid), and that $\|c-a\| + \|d-b\|\le \|a-b\| + \|c-d\|$ because $\{(a,b),(c,d)\}$ is a max-sum matching, the equation can be extended to 
\begin{align*}
	\|c-d\| + \|d-d'\| &  ~<~ \|c-a\| + \|a-d'\|\\
					   & ~=~ \|c-a\| + \|d-b\| \\
					   & ~\le~ \|a-b\| + \|c-d\|,
\end{align*}
which implies that $\|d-d'\|<\|a-b\|$. That is, $\|d-m\|=\|d-d'\|/2$ is smaller than the radius $\|a-m\|=\|a-b\|/2$ of $D_{ab}$, whose center is $m$. These observations imply the result.
\end{proof}

\begin{figure}[t]
	\centering
	\includegraphics[scale=1.0,page=3]{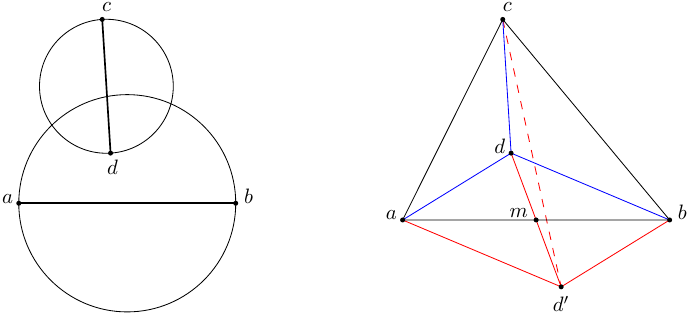}
	\caption{\small{Proof of Lemma~\ref{lem:point-in-disk}.}}	
	\label{fig:lem:point-in-disk}
\end{figure}

As mentioned above, we will elaborate a proof for $n=3$, that is, matchings of three segments. Since every pair of segments of a max-sum matching cross, or one of them points to the other one, we distinguish ten cases of relative position (or order-type) of the three segments, as shown in Figure~\ref{fig:cases}, enumerated from {\bf (A)} to {\bf (J)}. In the rest of this section, we will first show with direct proofs that in each of the cases from {\bf (A)} to {\bf (G)} the three disks have a common point. After that, we will use a proof by contradiction for each of the cases from {\bf (H)} to {\bf (J)}.

The general idea of these latter proofs by contradiction is the following: Since the disks pairwise intersect, but they do not have a common intersection, we can extend one of the segments (by moving one of its vertices) such that the new three disks have a singleton common intersection. In this new setting, the three new segments must fall again in one of the cases from {\bf (H)} to {\bf (J)}. Then, we show that the new segments must correspond to a max-sum matching given that the original ones do, but contradictorily such a new matching is not max-sum. 

\begin{figure}[t]
	\centering
	\includegraphics[scale=0.75,page=15]{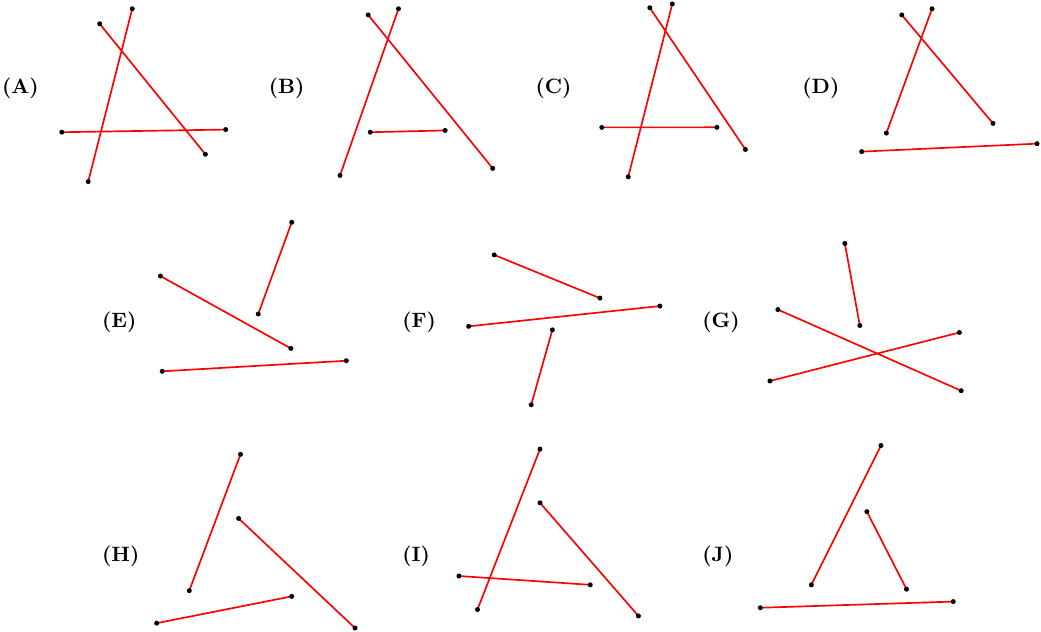}
	\caption{\small{The ten different relative positions of three segments.}}	
	\label{fig:cases}
\end{figure}

\begin{lemma}\label{lem:easy-cases}
If the segments of a max-sum matching of six points fall in one of the cases from {\bf (A)} to {\bf (G)}, then the three disks of the matching have a common intersection.
\end{lemma}

\begin{proof}
Let $\{a,b,c,a',b',c'\}$ be a 6-point set, and let $\M=\{(a,a'),(b,b'),(c,c')\}$ be a max-sum matching. In any case, refer to Figure~\ref{fig:lem-easy-cases} for the location of each point.

Case~{\bf (A)}: At least one altitude of the triangle $T$ bounded by the three segments goes through the interior of $T$. Let $u$ be the vertex of such an altitude in a side of $T$. By Thales' theorem, each of the three disks $D_{aa'}$, $D_{bb'}$, and $D_{cc'}$ contains $u$.

Case~{\bf (B)}: Let $u$ be the intersection point between $bb'$ and $cc'$. If $D_{bb'}$ contains $a$, then we are done since $a\in D_{cc'}$ because $\vec{a'a}$ points to $cc'$ (Lemma~\ref{lem:point-in-disk}). Similarly, if $D_{cc'}$ contains $a'$, then we are done since $a'\in D_{bb'}$ because $\vec{aa'}$ points to $bb'$. Otherwise, if $a\notin D_{bb'}$ and $a'\notin D_{cc'}$, then the triangle $\Delta aa'u$ is such that the interior angles at $a$ and $a'$, respectively, are both acute. Hence, the altitude $h$ from vertex $u$ goes through the interior of $\Delta aa'u$, and let $v\in aa'$ be the other vertex of $h$. Since $\vec{a'a}$ points to $cc'$, and $\vec{aa'}$ points to $bb'$, each of the disks $D_{aa'}$, $D_{bb'}$, and $D_{cc'}$ contains $v$, by Thales' theorem. The proof for Case~{\bf (C)} is both similar and simpler.

Case~{\bf (D)}: If $D_{bb'}$ contains $c'$, then we are done since $c'\in D_{aa'}$ because $\vec{cc'}$ points to $aa'$ (Lemma~\ref{lem:point-in-disk}). Similarly, if $D_{cc'}$ contains $b$, then we are done since $b\in D_{aa'}$ because $\vec{b'b}$ points to $aa'$. Otherwise, if $c'\notin D_{bb'}$ and $b\notin D_{cc'}$, then the triangle $\Delta c'bu$ is such that the interior angles at $c'$ and $b$, respectively, are both acute. Hence, the altitude $h$ from vertex $u$ goes through the interior of $\Delta c'bu$, and let $v\in c'b$ be the other vertex of $h$. We have $v\in D_{bb'}\cap D_{cc'}$, by Thales' theorem. Furthermore, since $c',b\in D_{aa'}$, we have that segment $c'b$ is contained in $D_{aa'}$. Hence, $v\in D_{aa'}\cap D_{bb'}\cap D_{cc'}$.

Cases~{\bf (E)},~{\bf (F)}, and~{\bf (G)}: In each of these cases, the same oriented segment points to each of the other two ones: Say, segment $\vec{aa'}$ points to both $bb'$ and $cc'$. Then, we have that $a'\in D_{bb'}\cap D_{cc'}$, by Lemma~\ref{lem:point-in-disk}. Hence, $a'\in D_{aa'}\cap D_{bb'}\cap D_{cc'}$. This completes the proof.
\end{proof}

\begin{figure}[t]
	\centering
	\includegraphics[scale=0.7,page=16]{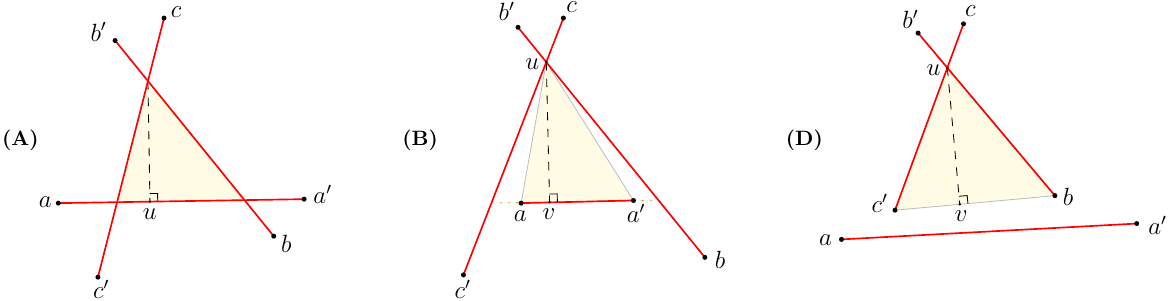}
	\caption{\small{Proof of Lemma~\ref{lem:easy-cases}.}}	
	\label{fig:lem-easy-cases}
\end{figure}

We now prove several technical lemmas, which will be used for the cases from {\bf (H)} to {\bf (J)}. The following lemma guarantee that if we extend one segment by moving one of the points, then the resulting segments correspond to a max-sum matching of the resulting point set. The rest of the lemmas impose monotone properties, and situations in which $\{(a,a'),(b,b'),(c,c')\}$ is not a max-sum matching of point set $\{a,b,c,a',b',c'\}$, all of them associated with cases from {\bf (H)} to {\bf (J)}.

\begin{lemma}\label{lem:extending}
Let $\M=\{(a_i,b_i),i=1,\dots,n\}$ be a max-sum matching of the set $P$ of $2n$ uncolored points, and let $c\notin P$ be a point such that $b_1$ belongs to the interior of the segment $a_1c$. Then, $\M^*=(\M\setminus\{(a_1,b_1)\})\cup\{(a_1,c)\}$ is a max-sum matching of $(P\setminus\{b_1\})\cup\{c\}$. 
\end{lemma}

\begin{proof}
Let $\M'$ denote any matching of $(P\setminus\{b_1\})\cup\{c\}$, and note that $(\M'\setminus\{(c,\M'(c))\})\cup\{(b_1,\M'(c))\}$ is a matching of $P$. Then,	
\begin{align*}
	\cost(\M') & ~=~ \cost\left(\M'\setminus\{(c,\M'(c))\}\right) + \|\M'(c)-c\|\\
	& ~\le~ \cost\left(\M'\setminus\{(c,\M'(c))\}\right) + \|\M'(c)-b_1\|+ \|b_1-c\| & \text{(triangle inequality)} \\
	& ~=~ \cost((\M'\setminus\{(c,\M'(c))\})\cup\{(b_1,\M'(c))\}) + \|b_1-c\|\\
	& ~\le~ \cost(\M) + \|b_1-c\| \\
	& ~=~ \cost(\M^*).
\end{align*}
Hence, the lemma follows.
\end{proof}

\begin{lemma}\label{lem:monotone-thing-1}
Let $p$, $p'$, $q$, and $q'$ be four points such that $\vec{pp'}$ points to $qq'$, and $q$ is to the right of $\ell(p,p')$. Let $z$ be a point to the left of both $\ell(p,p')$ and $\ell(q,q')$ such that: $(i)$ $q$ is to the left of $\ell(z,p)$; $(ii)$ vectors $p-z$ and $p'-z$ are orthogonal, and vectors $q-z$ and $q'-z$ are also orthogonal. Refer to Figure~\ref{fig:lem:monotone-thing-a}. Then, we have that
\begin{equation}\label{eq3}
\|p-z\|-\|q-z\| ~<~ \|p-q'\|-\|q-q'\|.
\end{equation}	
\end{lemma}

\begin{proof}
Rearranging terms in equation~\eqref{eq3}, we need to prove that
\begin{equation}\label{eq4}
	\|p-z\|+\|q-q'\| ~<~ \|p-q'\|+\|q-z\|.
\end{equation}
Conditions~$(i)$ and~$(ii)$, the fact that $\vec{pp'}$ points to $qq'$, and the location of $z$, imply that $q$ is to the right of $\ell(z,p')$ if and only if segments $pq'$ and $qz$ have a common point.

Suppose that $q$ is to the right of $\ell(z,p')$ (see Figure~\ref{fig:lem:monotone-thing-c}), case where segments $pq'$ and $qz$ have a common point. Then, $p$, $q$, $q'$, and $z$ are the vertices of a convex quadrilateral with non-empty interior and diagonals $pq'$ and $zq$. Hence, Equation~\eqref{eq4} holds by the triangle inequality.

Suppose now that $q$ is not to the right of $\ell(z,p')$ (see Figure~\ref{fig:lem:monotone-thing-b}), then segments $pq'$ and $qz$ do not intersect. Let $z'$ be the reflection of $z$ about the center of segment $qq'$, also the center of $D_{qq'}$. Then, we have $\|q-z\|=\|q'-z'\|$ and $\|q-q'\|=\|z-z'\|$, by condition~$(ii)$ and Thales' theorem. The fact that $p'$ belongs to the interior of $D_{qq'}$ given that $\vec{pp'}$ points to $qq'$, implies that $z'$ must be to the right of line $\ell(p,z)$. Then, since $pq'$ and $qz$ do not intersect, we have that $z$ belongs to triangle $\Delta pz'q'$. This implies
\[
	\|p-z\|+\|z-z'\|~<~\|p-q'\|+\|q'-z'\|,
\]
then equation~\eqref{eq3} and the result.
\end{proof}

\begin{figure}[t]
	\centering
	\begin{subfigure}[t]{0.5\textwidth}
		\centering
		\includegraphics[scale=1,page=6]{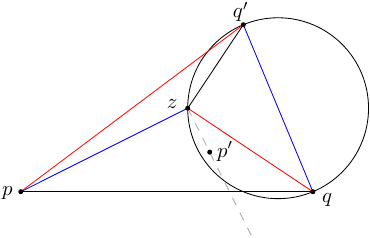}
		\caption{~}
		\label{fig:lem:monotone-thing-a}
	\end{subfigure}\hfill
	\begin{subfigure}[t]{0.5\textwidth}
		\centering
		\includegraphics[scale=1,page=7]{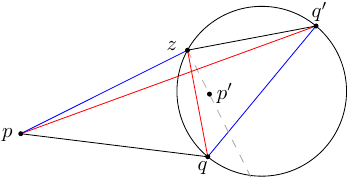}
		\caption{~}
		\label{fig:lem:monotone-thing-c}
	\end{subfigure}\\
	\begin{subfigure}[t]{0.5\textwidth}
		\centering
		\includegraphics[scale=1,page=5]{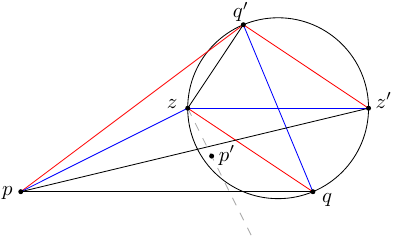}
		\caption{~}
		\label{fig:lem:monotone-thing-b}
	\end{subfigure}
	\caption{\small{
			(a) Statement of Lemma~\ref{lem:monotone-thing-1}.
			(b-c) Proof of Lemma~\ref{lem:monotone-thing-1}.
	}}
	\label{fig:lem:monotone-thing}
\end{figure}

\begin{lemma}\label{lem:monotone-thing-2}
Let $p$, $p'$, $q$, and $q'$ be four points in convex position such that $q$ and $q'$ are to the right and left of line $\ell(p,p')$, respectively. Let $z$ be a point to the left of both $\ell(p,p')$ and $\ell(q,q')$ such that: vectors $p-z$ and $p'-z$ are orthogonal, and vectors $q-z$ and $q'-z$ are also orthogonal. Then, we again have that $\|p-z\|-\|q-z\| < \|p-q'\|-\|q-q'\|$.
\end{lemma}

\begin{proof}
The proof is now simpler. In this case, segments $pq'$ and $qz$ have a common point, and the proof continues as that of Lemma~\ref{lem:monotone-thing-1}.
\end{proof}

\begin{lemma}\label{lem:common-point-1}
Let $a$, $b$, $c$, $a'$, $b'$, $c'$, and $z$ be seven points such that: $c$ is to the left of line $\ell(a,b)$; segments $\vec{aa'}$, $\vec{bb'}$, and $\vec{cc'}$ point to $bb'$, $cc'$, and $aa'$, respectively; and for each $u\in\{a,b,c\}$, point $z$ is to the left of line $\ell(u,u')$, and vectors $u-z$ and $u'-z$ are orthogonal. Refer to Figure~\ref{fig:lem:common-point-1}. Then, $\{(a,a'),(b,b'),(c,c')\}$ is not a max-sum matching of point set $\{a,b,c,a',b',c'\}$.
\end{lemma}

\begin{proof}
We will prove 
\begin{equation}\label{eq5}
	\|a-a'\|+\|b-b'\|+\|c-c'\| ~<~ \|a-b'\|+\|b-c'\|+\|c-a'\|.
\end{equation}
The conditions of the lemma guarantee (three times) the conditions of Lemma~\ref{lem:monotone-thing-1}. That is, we can apply Lemma~\ref{lem:monotone-thing-1} for $a$, $a'$, $b$, $b'$, and $z$ (where $a$ and $b$ play the role of $p$ and $q$, respectively); for $b$, $b'$, $c$, $c'$, and $z$ (where $b$ and $c$ play the role of $p$ and $q$, respectively); and for $c$, $c'$, $a$, $a'$, and $z$ (where $c$ and $a$ play the role of $p$ and $q$, respectively). Hence, we obtain the following three inequalities:
\begin{align*}
	\|a-z\|-\|b-z\| & ~<~ \|a-b'\|-\|b-b'\|, & (\text{Lemma~\ref{lem:monotone-thing-1}})\\
	\|b-z\|-\|c-z\| & ~<~ \|b-c'\|-\|c-c'\|, & (\text{Lemma~\ref{lem:monotone-thing-1}})\\
	\|c-z\|-\|a-z\| & ~<~ \|c-a'\|-\|a-a'\|. & (\text{Lemma~\ref{lem:monotone-thing-1}})
\end{align*}
Adding the above three inequalities, we obtain inequality~\eqref{eq5}. The result thus follows.
\end{proof}

\begin{figure}[t]
	\begin{subfigure}[t]{0.5\textwidth}
		\centering
		\includegraphics[scale=0.75,page=8]{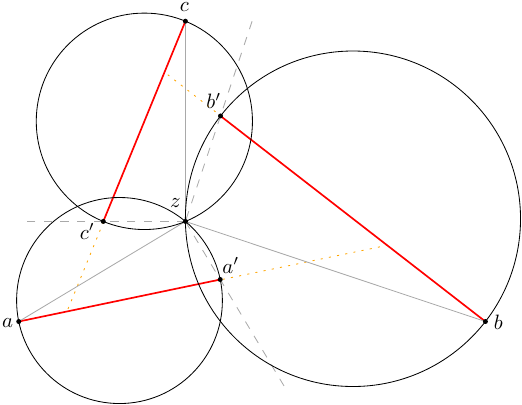}
		\caption{~}
		\label{fig:lem:common-point-1}
	\end{subfigure}\hfill
	\begin{subfigure}[t]{0.5\textwidth}
		\centering
		\includegraphics[scale=0.75,page=9]{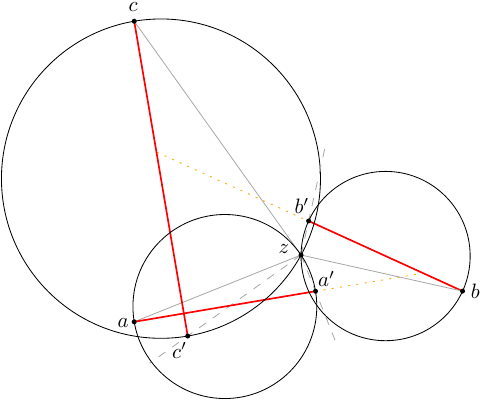}
		\caption{~}
		\label{fig:lem:common-point-2}
	\end{subfigure}
	\caption{\small{
			(a) Lemma~\ref{lem:common-point-1}.
			(b) Lemma~\ref{lem:common-point-2}.
	}}
\end{figure}

\begin{lemma}\label{lem:common-point-2}
Let $a$, $b$, $c$, $a'$, $b'$, $c'$, and $z$ be seven points such that: $c$ is to the left of line $\ell(a,b)$; segments $\vec{aa'}$ and $\vec{bb'}$ point to $bb'$ and $cc'$, respectively; segments $aa'$ and $cc'$ have a common point with $a$ and $a'$ to the right and left of line $\ell(c,c')$, respectively; and for each $u\in\{a,b,c\}$, point $z$ is to the left of line $\ell(u,u')$, and vectors $u-z$ and $u'-z$ are orthogonal. Refer to Figure~\ref{fig:lem:common-point-2}. Then, $\{(a,a'),(b,b'),(c,c')\}$ is not a max-sum matching of point set $\{a,b,c,a',b',c'\}$.
\end{lemma}

\begin{proof}
The conditions of the lemma guarantee (twice) the conditions of Lemma~\ref{lem:monotone-thing-1}. That is, we can apply Lemma~\ref{lem:monotone-thing-1} for $a$, $a'$, $b$, $b'$, and $z$ (where $a$ and $b$ play the role of $p$ and $q$, respectively); and for $b$, $b'$, $c$, $c'$, and $z$ (where $b$ and $c$ play the role of $p$ and $q$, respectively). Furthermore, we can apply Lemma~\ref{lem:monotone-thing-2} for $c$, $c'$, $a$, $a'$, and $z$ (where $c$ and $a$ play the role of $p$ and $q$, respectively). Hence, we obtain the following three inequalities:
\begin{align*}
	\|a-z\|-\|b-z\| & ~<~ \|a-b'\|-\|b-b'\|, & (\text{Lemma~\ref{lem:monotone-thing-1}})\\
	\|b-z\|-\|c-z\| & ~<~ \|b-c'\|-\|c-c'\|, & (\text{Lemma~\ref{lem:monotone-thing-1}})\\
	\|c-z\|-\|a-z\| & ~<~ \|c-a'\|-\|a-a'\|. & (\text{Lemma~\ref{lem:monotone-thing-2}})
\end{align*}
Adding the three inequalities above, we obtain again inequality~\eqref{eq5}. The result thus follows.
\end{proof}

Let $\alpha$ be a planar (open or closed) curve that splits the plane into two open regions. Given a point $p$ not in $\alpha$, let $H(\alpha,p)$ denote the region (between the two above ones) that contains $p$.  

\begin{proposition}\label{prop:hyperbola}
Let $a$, $b$, $a'$, and $b'$ be four points such that $\{(a,a'),(b,b')\}$ is a max-sum matching of $\{a,b,a',b'\}$. Let $\alpha$ be the arc of the hyperbola with foci $a$ and $b$ that goes through $b'$. Then, we have that $a'\in\alpha\cup H(\alpha,b)$. 
\end{proposition}

\begin{proof}
The arc $\alpha$ is the locus of the points $x$ of the plane such that $\|a-x\|-\|b-x\|=\|a-b'\|-\|b-b'\|$. Since $\{(a,a'),(b,b')\}$ is a max-sum matching, we have that $\|a-b'\|-\|b-b'\|\le \|a-a'\|-\|b-a'\|$, which implies the proposition.
\end{proof}

\begin{figure}[t]
	\begin{subfigure}[t]{0.32\textwidth}
		\centering
		\includegraphics[scale=0.75,page=13]{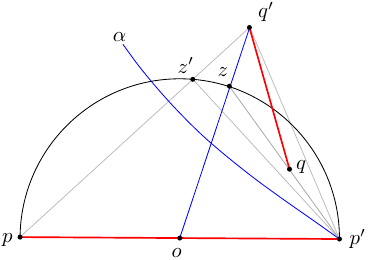}
		\caption{~}
		\label{fig:lem:monotone-thing-d}
	\end{subfigure}\hfill
	\begin{subfigure}[t]{0.32\textwidth}
		\centering
		\includegraphics[scale=0.75,page=12]{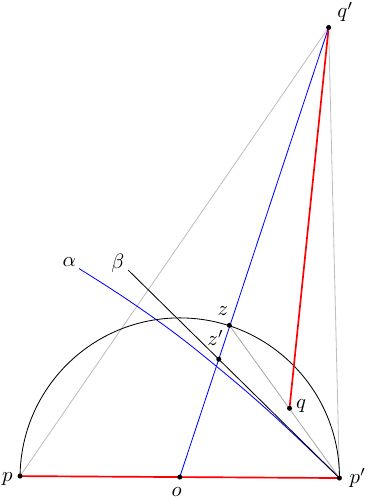}
		\caption{~}
		\label{fig:lem:monotone-thing-e}
	\end{subfigure}\hfill
	\begin{subfigure}[t]{0.3\textwidth}
		\centering
		\includegraphics[scale=0.8,page=14]{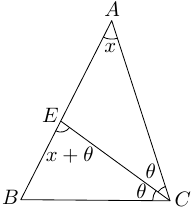}
		\caption{~}
		\label{fig:claim}
	\end{subfigure}	
	\caption{\small{
			(a,b,c) Proof of Lemma~\ref{lem:monotone-thing-3}.
	}}
\end{figure}

\begin{lemma}\label{lem:monotone-thing-3}
Let $p$, $p'$, and $o$ be three points such that $o$ is the midpoint of segment $pp'$. Let $z$ be a point of the circle $C_{pp'}$ to the left of line $\ell(p,p')$, $q$ a point of segment $zp'$ with $q\neq p'$, and $q'$ a point of ray $\tau(o,z)$ not in segment $oz$. Then, $\|p-p'\| + \|q-q'\| < \|p-q\| + \|p'-q'\|$. 
\end{lemma}

\begin{proof}
We divide the proof into two cases: $\|p-p'\|\ge \|q'-p'\|$; and $\|p-p'\|< \|q'-p'\|$. In both cases, let $\alpha$ be the arc of the hyperbola with foci $p$ and $q'$ that goes through $p'$.

In the first case (see Figure~\ref{fig:lem:monotone-thing-d}), let $z'$ be the intersection point of $C_{pp'}$ and $pq'$. Note that $\|p-p'\|\ge \|q'-p'\|$ implies that the region $H(\alpha,q')$ is convex. Furthermore, line $\ell(p',z')$ is perpendicular to the line $\ell(p,q')$ through the foci of $\alpha$, and this implies that $z'$ belongs to $H(\alpha,q')$, and then $z$ belongs to the convex intersection $H(\ell(p',z'),q')\cap H(\alpha,q')$. Given that $z$ belongs to the interior of $\alpha\cup H(\alpha,q')$, $p'$ is on the boundary of $\alpha\cup H(\alpha,q')$, and that $q\in zp'$, $q\neq p'$, we also have that $q$ belongs to $H(\alpha,q')$. This last fact, equivalent to $q\notin \alpha\cup H(\alpha,p)$, implies $\|q'-q\|-\|p-q\|<\|q'-p'\|-\|p-p'\|$, by Proposition~\ref{prop:hyperbola}, hence the result.

Consider now the second case, $\|p-p'\|< \|q'-p'\|$ (see Figure~\ref{fig:lem:monotone-thing-e}). Let $\beta$ be the bisector of the interior angle at $p'$ of triangle $\Delta op'q'$. By well-known properties of hyperbolas, $\beta$ is the tangent of $\alpha$ at $p'$. Furthermore, $\|p-p'\|< \|q'-p'\|$ implies that $\beta$ separates $\alpha$ from vertex $q'$. Let $z'$ be the intersection point between $oq'$ and $\beta$. We use the following claim to show that $\|o-z'\|<\|o-p'\|$: If $A$, $B$, and $C$ are the vertices of a triangle, and point $E$ belongs to side $AB$, such that line $\ell(C,E)$ is the bisector of the interior angle at $C$, then $\|C-B\|>\|B-E\|$. The claim follows from the fact that in any triangle, precisely in $\Delta BCE$, larger sides correspond to larger opposed interior angles (see Figure~\ref{fig:claim}). Applying the claim to $\Delta op'q'$, we have $\|o-z'\|<\|o-p'\|$, which implies that $\beta$ separates point $z$ and arc $\alpha$. Furthermore, $\beta$ separates point $q$ and $\alpha$ because $q\in zp'\setminus\{p'\}$. This last fact shows that $q$ is to left of $\alpha$ in the direction from $q'$ to $p$ (i.e.\ $q\notin \alpha\cup H(\alpha,p)$), implying $\|q'-q\|-\|p-q\|<\|q'-p'\|-\|p-p'\|$, by Proposition~\ref{prop:hyperbola}. The results thus follows.
\end{proof}

\begin{lemma}\label{lem:common-point-3}
Let $a$, $b$, $c$, $a'$, $b'$, $c'$, and $z$ be seven points such that: none of them is to the right of line $\ell(a,a')$; segments $\vec{b'b}$, $\vec{bb'}$, and $\vec{cc'}$ point to $aa'$, $cc'$, and $aa'$, respectively; and for each $u\in\{a,b,c\}$, point $z$ is to the left of line $\ell(u,u')$, and vectors $u-z$ and $u'-z$ are orthogonal. Refer to Figure~\ref{fig:lem:common-point-3}. Then, $\{(a,a'),(b,b'),(c,c')\}$ is not a max-sum matching of $\{a,b,c,a',b',c'\}$.
\end{lemma}

\begin{figure}[t]
	\centering
	\begin{subfigure}[t]{0.5\textwidth}
		\centering
		\includegraphics[scale=0.8,page=10]{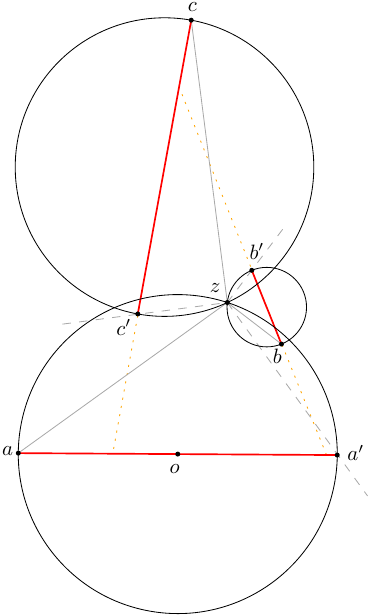}
		\caption{~}
		\label{fig:lem:common-point-3}
	\end{subfigure}\hfill
	\begin{subfigure}[t]{0.5\textwidth}
		\centering
		\includegraphics[scale=0.8,page=11]{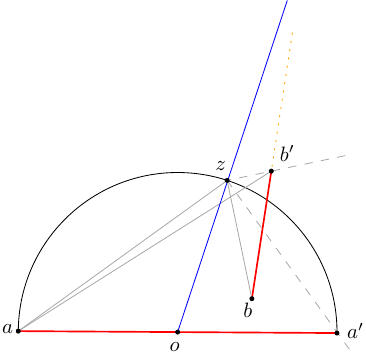}
		\caption{~}
		\label{fig:lem:common-point-4}
	\end{subfigure}	
	\caption{\small{
			(a,b) Proof of Lemma~\ref{lem:common-point-3}.
	}}
	\label{fig:lem:monotone-thing-2}
\end{figure}

\begin{proof}
We can apply Lemma~\ref{lem:monotone-thing-1} for $b$, $b'$, $c$, $c'$, and $z$ (where $b$ and $c$ play the role of $p$ and $q$, respectively); and for $c$, $c'$, $a$, $a'$, and $z$ (where $c$ and $a$ play the role of $p$ and $q$, respectively). Hence, we obtain the following two inequalities:
\begin{align*}
	\|b-z\|-\|c-z\| & ~<~ \|b-c'\|-\|c-c'\|, & (\text{Lemma~\ref{lem:monotone-thing-1}})\\
	\|c-z\|-\|a-z\| & ~<~ \|c-a'\|-\|a-a'\|. & (\text{Lemma~\ref{lem:monotone-thing-1}})
\end{align*}
Let $o$ denote the midpoint of segment $aa'$, also the center of $D_{aa'}$. Given that $z-c$ and $z-c'$ are orthogonal, and $\vec{cc'}$ points to $aa'$, we have that $c$ is to the left of line $\ell(o,z)$. Note that if $c$ is to the right of, or in, $\ell(o,z)$, then $c'$ would not belong to the interior of $D_{aa'}$. Similarly, since $z-b$ and $z-b'$ are orthogonal, and $\vec{b'b}$ points to $aa'$, we have that $b'$ is to the right of line $\ell(o,z)$. Then, since $\vec{bb'}$ points to $cc'$, we have that rays $\tau(b,b')$ and $\tau(o,z)$ must intersect.

If point $b$ belongs to triangle $\Delta aa'z$ (see Figure~\ref{fig:lem:common-point-4}), then the facts that $z-b$ and $z-b'$ are orthogonal, and $\tau(b,b')\cap \tau(o,z)\neq \emptyset$ imply that segments $bz$ and $ab'$ have a common point. Hence, $\|a-z\|+\|b-b'\|<\|b-z\|+\|a-b'\|$ by the triangle inequality. That is,
\[
	\|a-z\|-\|b-z\| ~<~ \|a-b'\|-\|b-b'\|.
\]
Adding the three inequalities above, we obtain again inequality~\eqref{eq5}, implying the result.

In the contrary case, $b$ does not belong to $\Delta aa'z$, (see Figure~\ref{fig:lem:common-point-3}), let us assume by contradiction that the matching $\{(a,a'), (b,b'),(c,c')\}$ is max-sum. Then, the matching $\{(a,a'),(b,b')\}$ is also max-sum. Let $w$ and $w'$ be the intersection points of $\ell(b,b')$ with $\tau(o,z)$ and $za'$, respectively. Note that $bb'\subset ww'$, and by Lemma~\ref{lem:extending}, $\{(a,a'),(w,w')\}$ is a max-sum matching of $\{a,a',w,w'\}$. But we can apply Lemma~\ref{lem:monotone-thing-3}, where $a$ and $w$ play the role of $p$ and $q$, respectively, to have that $\{(a,a'),(w,w')\}$ is not max-sum, which implies the result.
\end{proof}

We are now ready to give the proofs to the missing cases, those from {\bf (H)} to {\bf (J)}.

\begin{lemma}\label{lem:hard-cases}
If the segments of a max-sum matching of six points fall in one of the cases from {\bf (H)} to {\bf (J)}, then the three disks of the matching have a common intersection.
\end{lemma}

\begin{proof}
Suppose	by contradiction that the three disks, denoted $D_1$, $D_2$, and $D_3$, intersect pairwise, but without a common intersection (see Figure~\ref{fig:lem-hard-cases}). Let $u_{1,2}$, $u_{2,3}$, and $u_{3,1}$ be the vertices of the pairwise disjoint lenses $D_1\cap D_2$, $D_2\cap D_3$, and $D_3\cap D_1$, respectively, located in the triangle with vertices the centers of $D_1$, $D_2$, and $D_3$, respectively. 

\begin{figure}[t]
	\centering
	\includegraphics[scale=1.3,page=17]{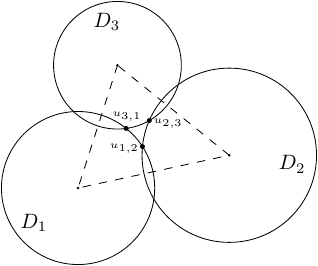}
	\caption{\small{Proof of Lemma~\ref{lem:hard-cases}.}}	
	\label{fig:lem-hard-cases}
\end{figure}

The idea is to use Lemma~\ref{lem:extending}, in combination with Lemmas~\ref{lem:common-point-1},~\ref{lem:common-point-2}, and~\ref{lem:common-point-3}, such that the point $z$ of these lemmas is among $u_{1,2}$, $u_{2,3}$, and $u_{3,1}$. To this end, we need to guarantee that point $z$ is not an extreme point of some segment of the matching. This is done in the next paragraph.

Note that two vertices among $u_{1,2}$, $u_{2,3}$, and $u_{3,1}$ cannot be the extreme points of a same segment of the matching. Furthermore, if each of the three vertices is an extreme point of some segment of the matching, then at least one pair of disjoint segments violates Lemma~\ref{lem:point-in-disk}. That is, the extreme point of one segment, in the interior of the convex hull of the four involved points, is not in the interior of the disk corresponding to the other segment. Hence, we can assume that at least one vertex among $u_{1,2}$, $u_{2,3}$, and $u_{3,1}$ is not an extreme point of a segment of the matching: say vertex $u_{1,2}$. This implies that we can extend the segment of disk $D_3$ by moving one of its extreme points such that the new three matching disks have a singleton common intersection at $u_{1,2}$. Let $z=u_{1,2}$, where $z$ is distinct from all the new six points. 

Let the new six points be denoted as $a$, $b$, $c$, $a'$, $b'$, and $c'$, in such a way that the new segments are precisely $aa'$, $bb'$, and $cc'$, and for each $u\in\{a,b,c\}$ point $z$ is to the left of line $\ell(u,u')$. By Lemma~\ref{lem:extending}, $\{(a,a'),(b,b'),(c,c')\}$ is a max-sum matching of $\{a,b,c,a',b',c'\}$.

It is important to note the following: If the original segments are in Case~{\bf (H)}, then by extending one segment we can stay in Case~{\bf (H)}, or go to Case~{\bf (I)}. If the original segments are in Case~{\bf (I)}, then by extending one segment we can stay in Case~{\bf (I)}, or go to Case~{\bf (C)} with a non-singleton common intersection of the three disks. Otherwise, if the original segments are in Case~{\bf (J)}, then by extending one segment we can stay in Case~{\bf (J)}, or go to Case~{\bf (B)} or~{\bf (D)} with a non-singleton common intersection of the three disks. Hence, since the common intersection of the new three disks $D_{aa'}$, $D_{bb'}$, and $D_{cc'}$ is singleton, we can ensure that the new segments $aa'$, $bb'$, and $cc'$ are again in a case from {\bf (H)} to {\bf (J)}, and the proof continues as follows.

If $aa'$, $bb'$, and $cc'$ fall in Case~{\bf (H)}, then by Lemma~\ref{lem:common-point-1} $\{(a,a'),(b,b'),(c,c')\}$ is not max-sum. If they are in in Case~{\bf (I)}, then by Lemma~\ref{lem:common-point-2} $\{(a,a'),(b,b'),(c,c')\}$ is not max-sum. Otherwise, if they are in in Case~{\bf (J)}, then by Lemma~\ref{lem:common-point-3} we have that $\{(a,a'),(b,b'),(c,c')\}$ is not max-sum. There exists a contradiction in each of the cases, thus the original three disks must have a common intersection. The lemma thus follows.
\end{proof}

Combining Lemma~\ref{lem:easy-cases} and Lemma~\ref{lem:hard-cases}, with Helly's theorem, we state now the main results of this paper:

\begin{theorem}
Let $P$ be a set of $2n$ (uncolored) points in the plane. Any max-sum matching $\M$ of $P$ is such that all disks of $D_{\M}$ have a common intersection.
\end{theorem}

\begin{corollary}
Let $P$ be a set of $2n$ (uncolored) points in the plane, and let $\{(a_i,b_i),i=1,\dots,n\}$ be a max-sum matching of $P$. Then, there exists a point $o$ of the plane such that for all $i\in\{1,\ldots,n\}$ we have:
\[
	\|a_i-o\|+\|b_i-o\|~\le~ \sqrt{2}\cdot  \|a_i-b_i\|.
\]
\end{corollary}

\blfootnote{
	\begin{minipage}[l]{0.3\textwidth}
		\includegraphics[trim=10cm 6cm 10cm 5cm,clip,scale=0.15]{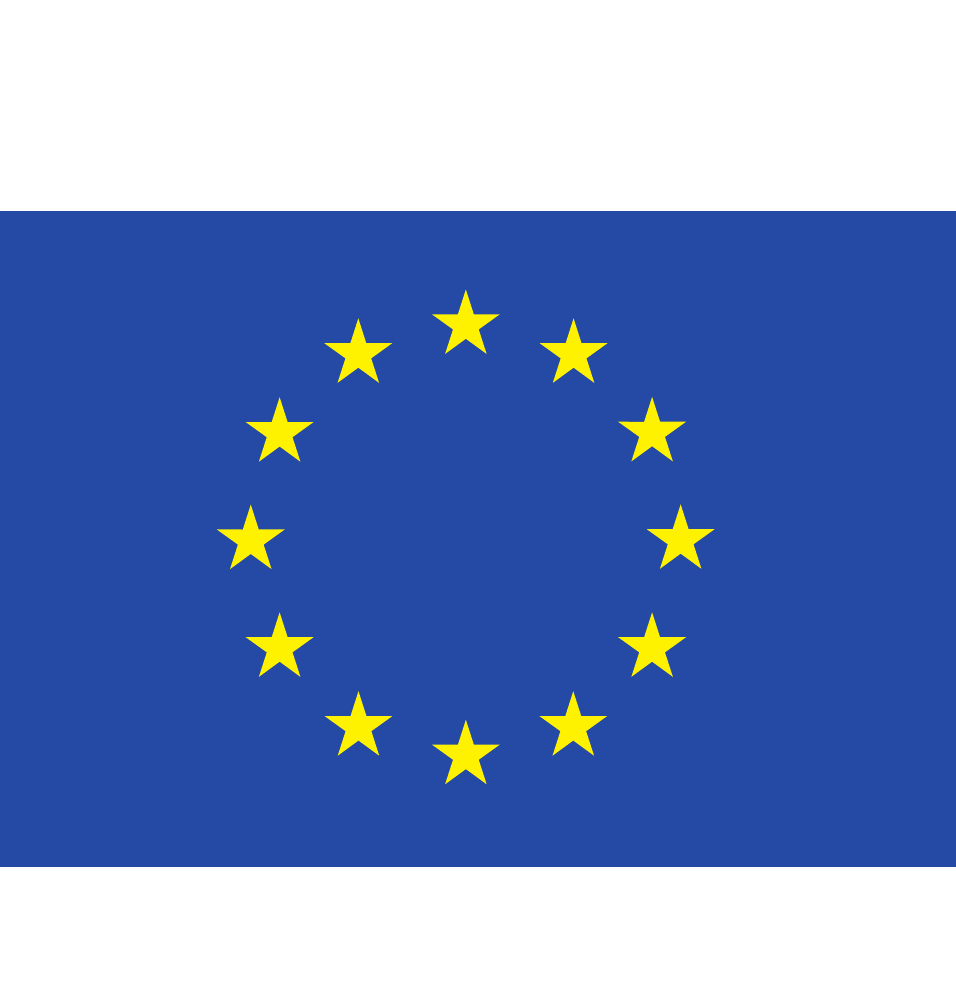}
	\end{minipage}
	\hspace{-2cm}
	\begin{minipage}[l][1cm]{0.75\textwidth}	
       	This work has received funding from the European Union's Horizon 2020
       	research and innovation programme under the Marie Sk\l{}odowska-Curie
       	grant agreement No 734922.
     \end{minipage}}

%
%
%


\bibliography{refs}

\end{document}